\documentclass{article}
\usepackage{amsmath,amssymb,amsfonts,hyperref}
\usepackage{graphicx,verbatim}

\usepackage{latexsym, amsmath, amsfonts, amssymb, amsthm, graphicx, epsfig, breqn, array, longtable, calc, multirow, hhline} 
\usepackage[lined,boxed,linesnumbered]{algorithm2e}

\def\tr{\mathrm{tr}}
\def\f{\mathbf f}
\def\u{\mathbf u}

\def\X{\mathcal X}
\def\I{\mathbf I}
\def\M{\mathbf M}

\def\B{\mathbf B}
\def\H{\mathbf H}
\def\K{\mathbf K}
\def\1{\mathbf 1}
\def\0{\mathbf 0}
\def\bbN{\mathbb N}
\def\bbR{\mathbb R}

\newtheorem{theorem}{Theorem}

\begin{document}

\title{Computing optimal experimental designs with respect to a compound Bayes risk criterion}
\author{Radoslav Harman$^{\ast,\dagger}$\let\thefootnote\relax\footnote{Department of Applied Mathematics and Statistics, Faculty of Mathematics, Physics and Informatics, Comenius University in Bratislava,  Slovakia}\let\thefootnote\relax\footnote{Department of Applied Statistics, Johannes Kepler University of Linz, Austria} and Maryna Prus$^\ddagger$\let\thefootnote\relax\footnote{Department of Mathematical Stochastics, Otto-von-Guericke-University, Magdeburg, Germany}}
\maketitle

\begin{abstract}
We consider the problem of computing optimal experimental design on a finite design space with respect to a compound Bayes risk criterion, which includes the linear criterion for prediction in a random coefficient regression model. We show that the problem can be restated as constrained A-optimality in an artificial model. This permits using recently developed computational tools, for instance the algorithms based on the second-order cone programming for optimal approximate design, and mixed-integer second-order cone programming for optimal exact designs. We demonstrate the use of the proposed method for the problem of computing optimal designs of a random coefficient regression model with respect to an integrated mean squared error criterion.
\bigskip

\bigskip


\end{abstract}

\section{Introduction}

Consider an experiment consisting of a set of trials. For each trial, it is possible to select a \emph{design point} $x$ from a finite \emph{design space} $\X$. 
\bigskip

We will formalize an \emph{exact design} of the experiment as a mapping $\xi: \X \to \bbN_0:=\{0,1,2,\ldots\}$, where $\xi(x)$, $x \in \X$, represents the number of trials to be performed in $x$\footnote{That is, we tacitly assume that all relevant aspects of the experimental design depend on the numbers of trials performed in individual design points, and not on the order of the trials.}. Therefore, the set of all exact designs is $\bbN_0^\X$. Accordingly, an \emph{approximate design} of the experiment is a mapping $\xi: \X \to \bbR_+:=[0,\infty)$, which we understand as a continuous relaxation of an exact design\footnote{Note that we do not assume that an approximate design is a normalized (probability) measure, which is usual in the optimal design literature; see \cite{HBF} for a justification.}. That is, the set of approximate designs is $\bbR_+^\X$.
\bigskip

An \emph{optimality criterion} is a function $\Phi: \bbR_+^\X \to \mathbb{R} \cup \{+\infty\}$ which quantifies exact or approximate designs in the sense that lower criterion values indicate a better design than larger criterion values\footnote{Note that the criterion may depend on an underlying statistical model. The criterion value of $+\infty$ means that the design is the worst possible, e.g., it does not permit an unbiased estimation of the parameters of interest of the model.}. If $\Xi \subset \bbR_+^\X$ is a set of permissible designs and
\begin{equation*}
\xi^* \in \mathrm{argmin}\{\Phi(\xi): \xi \in \Xi\},
\end{equation*}
we say that $\xi^*$ is a $\Phi$-optimal experimental design in the class $\Xi$. If $\Xi \subset \bbN_0^\X$, we say that $\xi^*$ is an \emph{optimal exact design}; if $\Xi$ is a convex set, we usually say that $\xi^*$ is an \emph{optimal approximate design}. Evidently, the existence and other properties of an optimal design depend on both $\Phi$ and $\Xi$.
\bigskip

Let $\f: \X \to \mathbb{R}^p$\footnote{The function $\f$ represents, for instance, the regression functions of a linear regression model, or gradients of the mean-value function of a non-linear regression model in accord with the approach of local optimality.}. For any approximate or exact design $\xi \in \bbR_+^\X$, let
\begin{equation*}
  \M(\xi)=\sum_{x \in \X} \xi(x) \f(x)\f(x)^\top
\end{equation*}
be the \emph{information matrix} of the design $\xi$. For many statistical models, the information matrix $\M(\xi)$ captures the amount of information about unknown parameters of interest, and the appropriate optimality criterion $\Phi$ can be expressed as a function of the information matrix. We will assume that $\mathrm{span}\{\f(x): \: x \in \X\} =\mathbb{R}^p$, which guarantees the existence of a design with a non-singular information matrix. See the monographs \cite{Pazman, Puk, Atkinson} for a more detailed introduction to optimal experimental design.

\subsection{The compound Bayes risk criterion}

The aim of this paper is to propose a method for computing optimal approximate and optimal exact designs with respect to the criterion $\Phi: \bbR_+^\X \to \mathbb{R} \cup \{+\infty\}$ defined by
\begin{equation}\label{eqn:CBRC}
  \Phi(\xi)=\sum_{j=1}^s \tr((\M(\xi)+\B_j)^{-1}\H_j),
\end{equation} 
for all $\xi \in \Xi$ such that $\M(\xi)+\B_j$ are non-singular for all $j=1,\ldots,s$, and defined by $\Phi(\xi)=+\infty$ for all other designs $\xi \in \Xi$. In \eqref{eqn:CBRC}, $\B_1,\ldots,\B_s$ are given non-negative definite $p \times p$ matrices, $\H_1,\ldots,\H_s$ are given positive definite $p \times p$ matrices. We will call \eqref{eqn:CBRC} the compound Bayes risk criterion (CBRC), because for $s=1$ it is the standard Bayes risk criterion; see \cite{gla} for details and \cite{pru2} for its relations to the prediction of individual deviations in random coefficient regression models.
\bigskip

Note also that a weighted-sum generalisation of the CBRC, i.e., if $\Phi_w(\xi)=\sum_j w_j \tr((\M(\xi)+\B_j)^{-1}\H_j)$, where $w_1,\ldots,w_s>0$ are given weights and $\M(\xi)+\B_j$ are non-singular for all $j=1,\ldots,s$, can be reduced to \eqref{eqn:CBRC} by the change $\H_j \to w_j \H_j$ for all $j=1,\ldots,s$.  
\bigskip 

For $s=1$, $\B_1=\0$ and $\H_1=\I_m$ (the identity matrix), the CBRC reduces to the criterion of $A$-optimality, which is one of the most common optimality criteria; see, e.g., the monographs cited above. It is also worth noting that if $\B_1$ is not $\0$, but it is the information matrix corresponding to an initial set of trials, then the minimization of the CBRC is equivalent to finding an $A$-optimal augmentation design of the set of initial trials. 

For $s=1$, $\B_1=\0$, and a general positive definite $\H_1$, CBRC is the linear optimality criterion. Because the linear criterion is convex on the set of non-negative definite symmetric matrices, (e.g., \cite{Pazman}), the standard theorems of convex analysis imply that the general CBRC is also convex.  

For $s=2$, $\B_1=\0$, a general positive definite $\B_2$, and $\H_1$ a multiple of $\H_2$, the CBRC is the recently proposed criterion for a prediction of individual parameters in random coefficient regression models; see \cite{pru3}. This special case of the CBRC is our main motivation for studying CBR criteria; see the next subsection for details.
   
\subsection{A random coefficient regression model}\label{subsec:RCM}
Random coefficient regression (RCR) models are popular in many fields of statistical application, especially in biosciences and in medical research. In these models observational units (individuals) are assumed to come from the same population and differ from each other by individual random parameters. The problem of design optimization for estimation of a population parameter is well discussed in the literature (e.g., \cite{fed} or \cite{ent}).
\bigskip

The models with \emph{known} population parameters have been investigated by \cite{gla}. It has been established that Bayes optimal designs are optimal for the prediction of the individual parameters. 
 For models with \textit{unknown} population parameters, sufficient and necessary conditions\footnote{The so-called equivalence theorem.} for optimal approximate designs for the prediction of individual parameters  are presented in \cite{pru3}. However, the closed form of the optimal approximate designs is available only for some special cases, and the problem of the numerical computation of approximate and exact designs in general RCR model remained open. 
\bigskip

In an RCR model the $j$-th trial (observation) ${Y}_{ij}$ of individual $i$ is given by ${Y}_{ij}=\mathbf{f}(x_{ij})^\top \mbox{\boldmath{$\beta $}}_i+ \varepsilon_{ij}$ for $j=1,\ldots,m_i$ and $i=1,\ldots,n$, where $m_i$ is the number of trials at individual $i$, $n$ is the number of individuals, and $\mathbf{f} : \X \to \mathbb{R}^p$ are known regression functions. We assume that the experimental settings $x_{ij}$ may range over the \textit{finite} design space $\X$\footnote{This can be a practical finite discretization of a theoretical continuous design space.}. The observational errors $\varepsilon_{ij}$ have zero mean and a common variance $\sigma^2>0$. The individual parameters $\mbox{\boldmath{$\beta $}}_i=( \beta_{i1}, \ldots, \beta_{ip})^\top$ are realizations from a common distribution with unknown population mean $\mathrm{E}\,(\mbox{\boldmath{$\beta $}}_i)={\mbox{\boldmath{$\beta $}}}$ and a population covariance matrix $\mathrm{Cov}\,(\mbox{\boldmath{$\beta $}}_i)=\sigma^2\mathbf{D}$ for a given non-singular $p\times p$ matrix $\mathbf{D}$. All individual parameters $\mbox{\boldmath{$\beta $}}_{i}$ and all observational errors $\varepsilon_{ij}$ are assumed to be uncorrelated.
\bigskip

We work with the particular random coefficient regression models, in which all individuals are observed under the same regime, i.\,e.,\  all individuals $i=1,\ldots,n$ have the same number $m_i=m$ of trials at the same design points $x_{ij}=x_j$:
\begin{equation}\label{2}
	{Y}_{ij}=\mathbf{f}(x_{j})^\top \mbox{\boldmath{$\beta $}}_i+ \varepsilon_{ij},
\end{equation}
for $j=1, \ldots, m$.
\bigskip

Under the exact design $\xi$, linear unbiased predictors of $\mbox{\boldmath{$\beta$}}_i$'s exist if and only if $\M(\xi)$ is non-singular (see \cite{pru1}, Ch.~4). Let $\hat{\mbox{\boldmath{$\beta$}}}_i(\xi)$ be the best linear unbiased predictor of $\mbox{\boldmath{$\beta$}}_i$ for all $i=1,\ldots,n$. If $\M(\xi)$ is non-singular, the linear criterion for the prediction of the individual parameters $\mbox{\boldmath{$\beta $}}_i$ in the model (\ref{2}) is the sum across all individuals of the traces of the mean squared error matrices for linear combinations of the individual random parameters (see \cite{pru3}), i.e.,
\begin{equation}\label{lin}
\sum_{i=1}^{n}\mathrm{tr}\,\left(\mathrm{Cov}\,\left(\mathbf{L}^\top \hat{\mbox{\boldmath{$\beta $}}}_i(\xi)-\mathbf{L}^\top \mbox{\boldmath{$\beta $}}_i\right)\right),
\end{equation}
 where $\mathbf{L}$ is some $p\times q$ matrix. For $\mathbf{A}=\mathbf{L}\mathbf{L}^\top$ the criterion (\ref{lin}) results in
\begin{equation}\label{lreg}
\mathrm{L}_{\rm pred}(\xi)= \mathrm{tr}\,(\mathbf{M}(\xi)^{-1}\mathbf{A})+(n-1)\,\mathrm{tr}\,((\mathbf{M}(\xi)+\mathbf{D}^{-1})^{-1}\mathbf{A}).
\end{equation}
For a non-singular matrix $\mathbf{A}$ (full row rank matrix $\mathbf{L}$) the criterion (\ref{lreg}) may be recognized as a particular case of the CBRC (\ref{eqn:CBRC}) for $s=2$, $\B_1=\0$, $\B_2=\mathbf{D}^{-1}$, $\H_1=\mathbf{A}$ and $\H_2=(n-1)\,\mathbf{A}$.

The IMSE-criterion for the prediction is defined as the sum of integrated mean squared distances between the predicted and real individual response with respect to a suitable measure $\nu$ on $\mathcal{X}$:
\begin{equation}\label{imse}
\sum_{i=1}^{n}\mathrm{E}\,\left(\int_{\mathcal{X}} (\mathbf{f}(x)^\top \hat{\mbox{\boldmath{$\beta$}}}_i(\xi) - \mathbf{f}(x)^\top \mbox{\boldmath{$\beta $}}_i)^2 \nu (\mathrm{d}x)\right).
\end{equation}
The criterion results in
\begin{equation}\label{phi}
\mathrm{IMSE}_{\rm pred}(\xi) = \mathrm{tr}\,(\mathbf{M}(\xi)^{-1}\mathbf{V})+(n-1)\mathrm{tr}\,((\mathbf{M}(\xi)+\mathbf{D}^{-1})^{-1}\mathbf{V}) \, ,
\end{equation}
where $\mathbf{V}=\int_{\mathcal{X}} \mathbf{f}(x)\mathbf{f}(x)^\top\nu (\mathrm{d}x)$, and can be recognized as a particular case of the linear criterion (\ref{lreg}) for $\mathbf{A}=\mathbf{V}$. A sufficient condition for the matrix $\mathbf{A}$ to be positive definite is that $\nu$ has positive values at all points of $\X$. 


\subsection{Contribution of the paper}

In principle, optimal approximate and efficient exact designs with respect to the CBRC can be computed by methods applicable to general convex criteria; see, e.g., \cite{man}.  However, the non-standard nature of the CBRC requires specific implementation of the general methods, these methods tend to be  slow, and it is usually difficult to modify them to handle non-standard constraints on the design.
\bigskip

The main idea of this paper is to compute optimal approximate and optimal exact designs for the CBRC by utilizing an artificial optimal design problem, with respect to the standard criterion of $A$-optimality, but with an extended design space and additional linear constraints. This permits using the modern methods of mathematical programming developed for $A$-optimality, such as second-order cone programming, mixed-integer second-order cone programming (see \cite{SagnolHarman}), or integer quadratic programming (analogously to \cite{HF}) to solve the problem of CBR-optimal designs, even under additional linear constraints on the design. For instance, this approach allows computing the CBR-optimal replication-free exact designs, which is often relevant for applications (e.g., \cite{ras}, cf. Section 2.11 of \cite{fedorovleonov}). Note that the mathematical programming methods for $A$-optimality under linear constraints are implemented in the freely available package \texttt{OptimalDesign} for the computing environment \texttt{R} (see \cite{OD}).
\bigskip

In an abstract sense, the idea of this paper is opposite to the technique of the conversion of a constrained design problem to a compound design problem (e.g., \cite{mik} and \cite{coo}), which, however, requires a non-trivial identification of a Lagrange multiplier involved in the compound criterion. In contrast, the method proposed in this paper is completely straightforward, and involves a different kind of constraints.
\bigskip

In Section \ref{sec:main}, we formulate the main theorem, which enables the conversion of the problem of CBR-optimality to the problem of constrained $A$-optimality. In Section \ref{sec:example}, we demonstrate the use of the method for computing optimal designs of a random coefficient regression model with respect to the integrated mean squared error criterion. In particular, we show that our approach can be used to compute CBR-optimal exact designs with respect to constraints that guarantee replication-free designs with a given minimum time delay between consecutive trials.

\section{Conversion of CBR-optimality to constrained A-optimality}\label{sec:main}

Let $r_1,\ldots,r_s$ be the ranks of $\B_1,\ldots,\B_s$. Let $\tilde{\mathcal{Y}}_1=\{y^{(1)}_1,\ldots,y^{(1)}_{r_1}\},\ldots,\tilde{\mathcal{Y}}_s=\{y^{(s)}_1,\ldots,y^{(s)}_{r_s}\}$ be auxiliary sets. For all $j \in \{1,\ldots,s\}$, let $\tilde{\X}_j=\{j\} \times \X$, where $\X$ is the original set of design points. We will also assume that the auxiliary sets are chosen such that $\tilde{\mathcal{X}}_1,\ldots,\tilde{\mathcal{X}}_s, \tilde{\mathcal{Y}}_1,\ldots,\tilde{\mathcal{Y}}_s$ are mutually disjunctive.  The set of design points for our artificial model will be $\tilde{\X}:=\cup_{j=1}^s (\tilde{\X}_j \cup \tilde{\mathcal{Y}}_j)$. 
\bigskip

Let $j \in \{1,\ldots,s\}$ and let $\K_j \in \mathbb{R}^{p \times p}$ be such that $\H_j=\K_j\K^\top_j$. Because the matrices $\H_j$ are positive definite, the matrices $\K_j$ are non-singular. Next, for all $x \in \X$ and $j=1,\ldots,s$ define $\tilde{\f}_j(x):=\K_j^{-1}\f(x)$ and let
\begin{equation*}
\K_j^{-1}\B_j\K_j^{-\top}=\sum_{k=1}^{r_j} \u_j(y^{(j)}_k)\u_j(y^{(j)}_k)^{\top}
\end{equation*}
for some $\u_j(y^{(j)}_k) \in \mathbb{R}^p$, $k=1,\ldots,r_j$.
\bigskip

Next, let $\tilde{\f}: \tilde{\X} \to \mathbb{R}^{sp}$ be defined as
\begin{equation*}
\tilde{\f}(\tilde{x})=\left\{
\begin{array}{ll}
  (\0_p^\top,\ldots,\0_p^\top,\tilde{\f}_j(x)^\top,\0_p^\top,\ldots,\0_p^\top)^\top \in \bbR^{sp} \text{ if } \tilde{x}=(j,x) \in \tilde{\X}_j \\
  (\0_p^\top,\ldots,\0_p^\top,\u_j(y)^\top,\0_p^\top,\ldots,\0_p^\top)^\top \in \bbR^{sp} \text{ if } \tilde{x}=y \in \tilde{\mathcal{Y}}_j
\end{array}
\right.
\end{equation*}
for all $j \in \{1,\ldots,s\}$, $x \in \X$, and $y \in \tilde{\mathcal{Y}}_j$, where the blocks $\tilde{\f}_j(x)^\top$ and $\u_j(y)^\top$ in the above expressions form the $j$-th $p$-dimensional blocks within the $(sp)$-dimensional vectors. The vectors $\tilde{\f}$ will form the regression functions of the artificial model on $\tilde{\X}$. Let $\Xi$ be a set of permissible (approximate or exact) designs on $\X$. To avoid uninteresting cases, we will assume that $\Xi$ contains a design $\xi$ such that $\Phi(\xi)<+\infty$. Recall that our primary aim is to find a solution $\xi^*$ of the problem 
\begin{eqnarray}\label{primaryOptProblem}
\min && \Phi(\xi) \nonumber \\
\text{s.t.} && \xi \in \Xi, \nonumber 
\end{eqnarray}
where $\Phi$ is the CBRC defined by \eqref{eqn:CBRC}.
\bigskip

Note that if $\tilde{\xi}$ is a design on $\tilde{\X}$, then its part $\tilde{\xi}(1,\cdot): x \to \tilde{\xi}(1,x)$ is a design on $\X$. For a design $\tilde{\xi}$ on $\tilde{\X}$ with a non-singular information matrix
\begin{equation*}
\tilde{\M}(\tilde{\xi})=\sum_{\tilde{x} \in \tilde{\X}} \tilde{\xi}(\tilde{x}) \tilde{\f}(\tilde{x})\tilde{\f}(\tilde{x})^\top
\end{equation*}
let $\Phi_A(\tilde{\xi})=\mathrm{tr}\left(\tilde{\M}^{-1}(\tilde{\xi})\right)$, and let $\Phi_A(\tilde{\xi})=+\infty$ if $\tilde{\M}(\tilde{\xi})$ is singular. That is, $\Phi_A$ is the standard criterion of $A$-optimality.
\begin{theorem}\label{thm:main}
Let $\tilde{\xi}^*$ be a solution of the optimization problem
\begin{eqnarray}\label{mainOptProblem}
\min && \Phi_A(\tilde{\xi}) \\
\text{s.t.} && \tilde{\xi}(j,x) = \tilde{\xi}(1,x) \text{ for all } x \in \X \text{ and } j=2,\ldots,s, \nonumber \\
&& \tilde{\xi}(y) = 1 \text{ for all } y \in \tilde{\mathcal{Y}}_j \text{ and } j=1,\ldots,s, \nonumber \\
&& \tilde{\xi}(1,\cdot) \in \Xi. \nonumber 
\end{eqnarray}
Then $\tilde{\xi}^*(1,\cdot)$ is a CBR-optimal design in the class $\Xi$.
\end{theorem}
\begin{proof}
  Let us use $\tilde{\Xi}$ to denote the set of all designs $\tilde{\xi}$ satisfying the constraints of \eqref{mainOptProblem}. First, we will prove the following key fact: \begin{equation}\label{KeyFact}
  \text{For all } \tilde{\xi} \in \tilde{\Xi} \text{ we have } \Phi_A(\tilde{\xi})=\Phi(\tilde{\xi}(1,\cdot)),
\end{equation}  
  where $\Phi$ is the CBRC. For any $\tilde{\xi} \in \tilde{\Xi}$ it is straightforward to verify that $\tilde{\M}(\tilde{\xi})$ is a block-diagonal matrix with $s$ diagonal $p \times p$ blocks
\begin{equation}\label{diagonalBlocks}
  \K_j^{-1}\left(\M(\tilde{\xi}(1,\cdot))+\B_j\right)\K_j^{-T}, \:\: j=1,\ldots,s.
\end{equation}
Evidently, $\Phi_A(\tilde{\xi})=+\infty$ if and only if at least one of the matrices \eqref{diagonalBlocks} is singular, which is if and only if $\Phi(\tilde{\xi}(1,\cdot))=+\infty$. If $\Phi_A(\tilde{\xi}) < +\infty$, then all matrices in  \eqref{diagonalBlocks} are non-singular, which means that 
\begin{eqnarray*}
  \Phi_A(\tilde{\xi})&=&\mathrm{tr}\left(\tilde{\M}^{-1}(\tilde{\xi})\right)=\sum_{j=1}^s \mathrm{tr}\left(\K_j^T\left(\M(\tilde{\xi}(1,\cdot))+\B_j\right)^{-1}\K_j\right)=\\
  &=&\sum_{j=1}^s \mathrm{tr}\left(\left(\M(\tilde{\xi}(1,\cdot))+\B_j\right)^{-1}\H_j\right)=\Phi(\tilde{\xi}(1,\cdot)).
\end{eqnarray*}
Observe also that 
\begin{equation}\label{Observation}
 \text{For each } \xi \in \Xi \text{ there is exactly one } \tilde{\xi} \in \tilde{\Xi} \text{ such that } \xi=\tilde{\xi}(1,\cdot).
\end{equation} 
 Now, assume that $\xi^c \in \Xi$ is a candidate permissible design for the original design problem and let $\tilde{\xi}^* \in \tilde{\Xi}$ be any solution of \eqref{mainOptProblem}. Using Observation \eqref{Observation}, Fact \eqref{KeyFact}, optimality of $\tilde{\xi}^*$ in $\tilde{\Xi}$, and again Fact \eqref{KeyFact}, we obtain
\begin{equation*}
  \Phi(\xi^c)=\Phi(\tilde{\xi}^c(1,\cdot))=\Phi_A(\tilde{\xi}^c) \geq \Phi_A(\tilde{\xi}^*) = \Phi(\tilde{\xi}^*(1,\cdot)).  
  \end{equation*}
This proofs that $\tilde{\xi}^*(1,\cdot)$ is the optimal design for the original problem.
\end{proof}

Thus, Theorem \ref{thm:main} allows us to convert the problem of CBR-optimality to a problem of $A$-optimality for an artificial model with an extended design space $\tilde{\X}$ with additional linear constraints. For linearly constrained $A$-optimal design problems, we can use recently developed theoretical characterisations and computational tools; see the examples in the next section.

\section{Example: optimal designs for a random coefficient model}\label{sec:example}

In this section we consider the straight line regression model
\begin{equation}\label{lr}
	Y_{ij}= \beta_{i1}+\beta_{i2}x_j+\varepsilon_{ij}\,,
\end{equation}
$i=1, \ldots, n$, $j=1, \ldots, m$, as a special case of the model (\ref{2}) with the regression function $\mathbf{f}(x)=(1,x)^\top$, on the design space $\X=\{k/(d-1): k = 0,1,\ldots,d-1\}$, where $d \geq 2$.

For numerical calculations we fix the number of individuals and the number of trials per individual by $n=100$ and $m=10$, respectively, and the size of the design space by $d=51$. We assume a diagonal structure of the covariance matrix of random parameters, i.e. the random intercept and the random slope are uncorrelated with each other: $\mathbf{D}= \mathrm{diag}(\delta_1,\delta_2)$. We fix the intercept dispersion by the small value $\delta_1=0.01$. 
For the IMSE-criterion we consider the uniform measure $\nu(x)=1/d$, for all $x\in \X$, which implies $\mathbf{A}=\frac{1}{d}\sum_{l=1}^{d}\mathbf{f}(x_l)\mathbf{f}(x_l)^\top$.
\bigskip

It was established in \cite{pru3} that approximate optimal designs in this model have only two support points $x=1$ and $x=0$. Therefore, only the optimal number of trials $m_1^*$ at $x=1$ has to be determined.

Further we investigate the behaviour of optimal designs in dependence of the slope variance. Instead of the variance parameter $\delta_2 \in (0, \infty)$ we use the rescaled slope variance $\rho=\delta_2/(1+\delta_2)$, which is monotonically increasing in $\delta_2$ and takes all its values on the interval $(0,1)$.
\bigskip

We use the procedure \texttt{od.SOCP} from the package \texttt{OptimalDesign} in \texttt{R} for computing optimal approximate designs and procedure \texttt{od.MISOCP} for computing optimal exact designs (see \cite{SagnolHarman} for the theoretical background). 
\bigskip

Figure~\ref{f1} illustrates the behaviour of the optimal number of trials  $m_1^*$ for approximate (dashed line) and exact (solid line) designs. 

\begin{figure}[ht]
\centering
       \includegraphics[width=100mm]{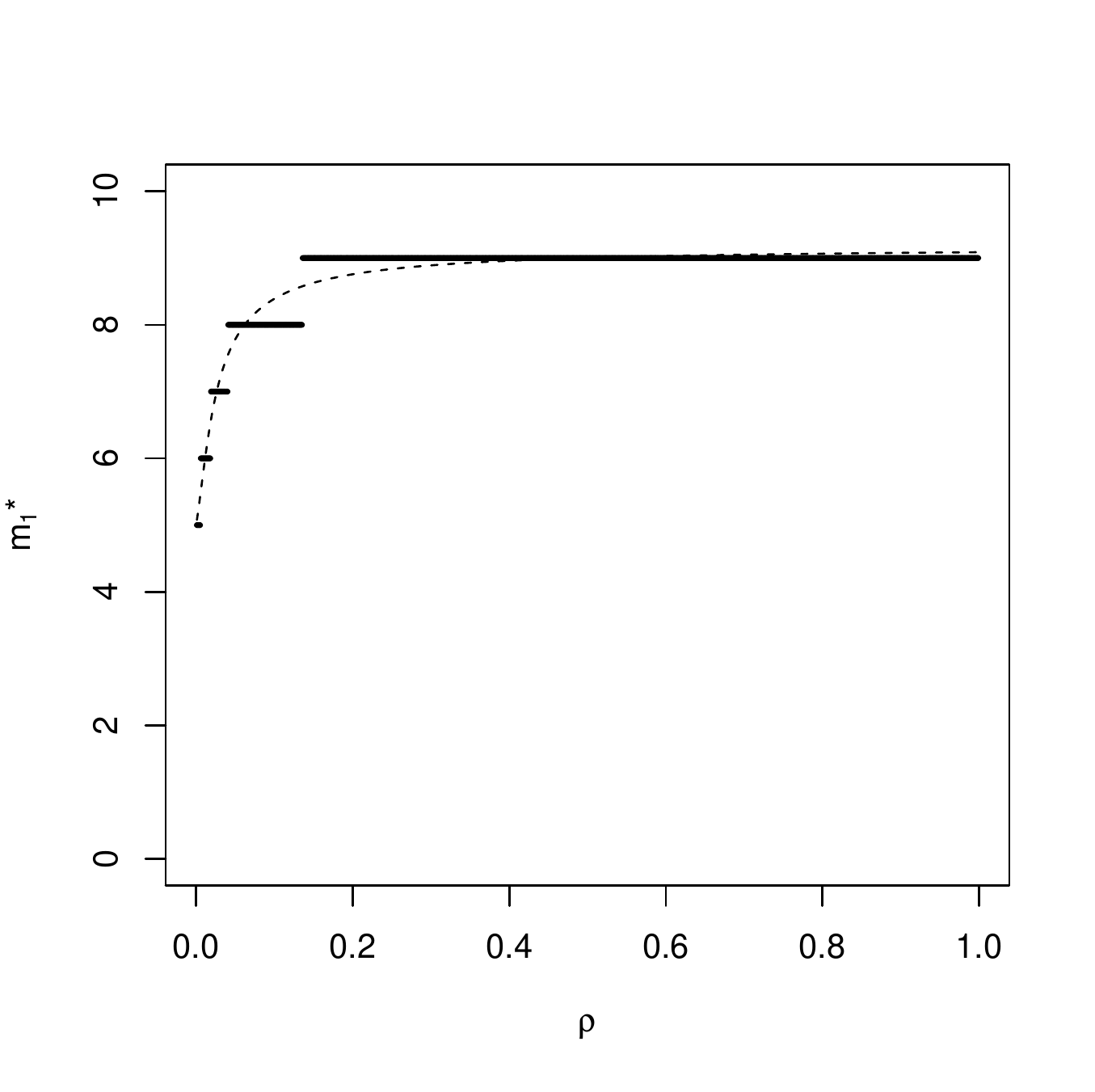}
       \caption{The number $m_1^*$ of trials in the design point $x=1$ for the IMSE-optimal exact (solid line) and approximate (dashed line) designs in model \eqref{lr} (without additional constraints) with respect to the values of the rescaled slope variance $\rho$}
       \label{f1}
\end{figure}  
		
Both exact and approximate optimal designs on the figure start at $m_1^*=5$ for small values of the parameter $\rho$. This is an expected result since for small values of $\delta_2$ the model becomes very close to the fixed effects model, for which balanced designs are optimal. Then the optimal number of trials $m_1^*$ increases with $\rho$.

The values of optimal exact designs on different intervals for $\rho$ are given in Table~\ref{tab1}. Note that the optimal designs do not change within a relatively long intervals of $\rho$, which allows for a certain degree of error in the guess of the nominal value of the variance parameter $\delta_2$.

\begin{table}
  \centering
\begin{tabular}{ ccc }
\hline
 $\rho$ & $\xi^*(0)$ & $\xi^*(1)$ \\
\hline
  $(0, 0.005]$ & 5 & 5 \\
	$[0.006, 0.018]$ & 4 & 6 \\
	$[0.019, 0.040]$ & 3 & 7 \\
	$[0.041, 0.135]$ & 2 & 8 \\
	$[0.136, 1)$ & 1 & 9 \\
\hline
 \end{tabular}
  \caption{IMSE-optimal exact designs $\xi^*$ in model \eqref{lr} (without additional constraints) with respect to the values of the rescaled slope variance $\rho$}\label{tab1}
\end{table}



\bigskip

Now we will consider model (\ref{lr}) with the following additional constraints: within every triple of neighbouring design points, only one trial is possible, i.\,e. permissible designs $\xi$ satisfy (we assume $d \geq 3$ here)
\begin{equation}\label{const}
\xi\left(\frac{k}{d-1}\right)+\xi\left(\frac{k+1}{d-1}\right)+\xi\left(\frac{k+2}{d-1}\right) \leq 1, \quad k=0, \ldots, d-3.
\end{equation}
This may correspond to the requirement of a minimum time distance between successive observations on the same subject.
\bigskip

 Note that permissible designs may take observations at all points of the design region in this case. The next graphics (Figure~\ref{f2}) illustrate the behaviour of optimal exact designs in model with constraints \eqref{const}. These designs take one observation at each support point. 
\begin{figure}[ht]
\centering
       \includegraphics[width=100mm]{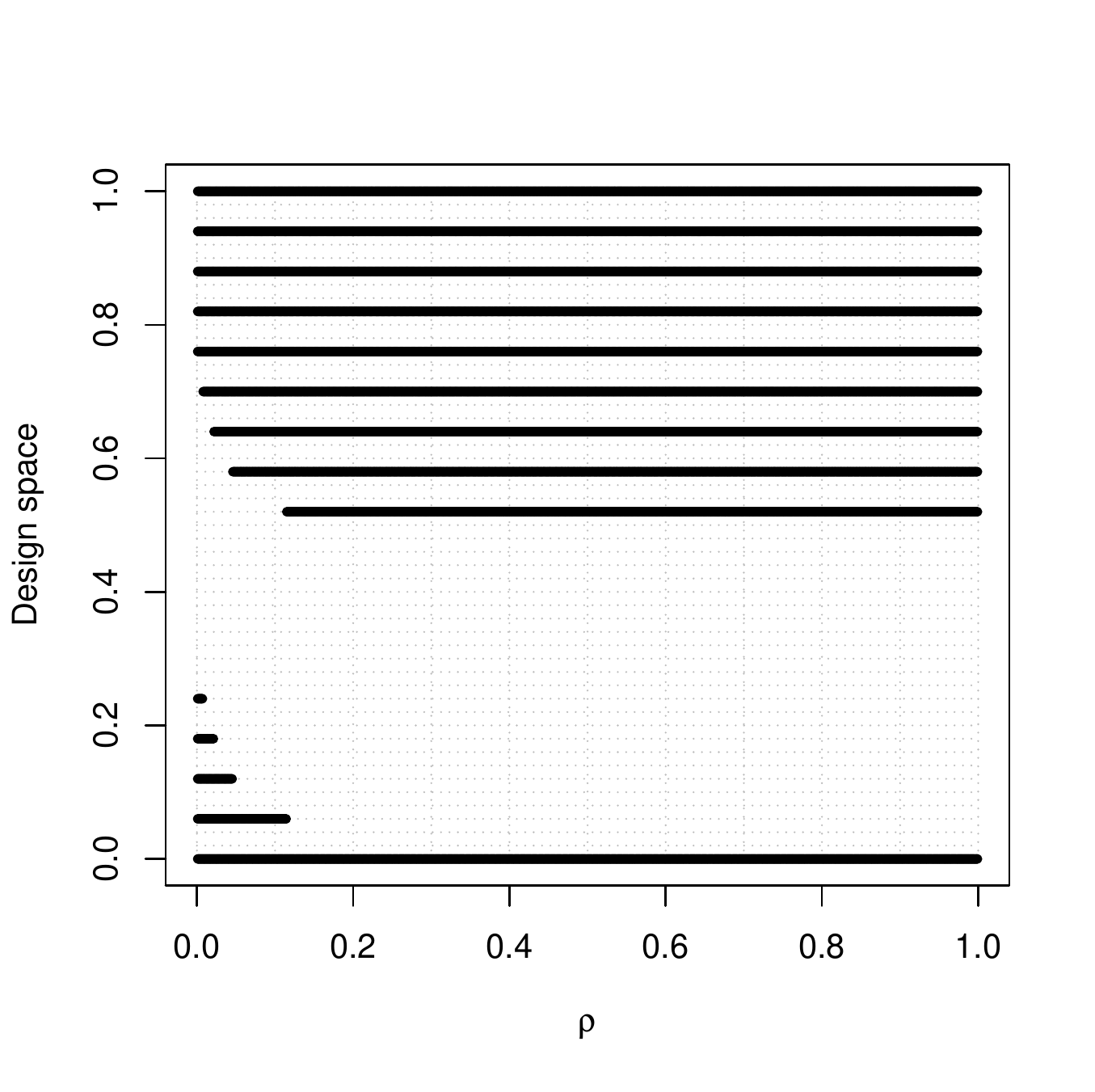}
       \caption{IMSE-optimal exact designs in model \eqref{lr} with additional constraints \eqref{const} with respect to the values of the rescaled slope variance $\rho$}\label{f2}      
    \end{figure}   		
The support points are given by the next table (Table~\ref{tab2}) with respect to the values of the rescaled slope variance $\rho$.
\begin{table}
  \centering
\begin{tabular}{ccccccccccc}
\hline
	$\rho$ & $x_1$ & $x_2$ & $x_3$ & $x_4$ & $x_5$ & $x_6$ & $x_7$ & $x_8$ & $x_9$ & $x_{10}$\\
\hline
 $(0, 0.007]$  & $0$ &$0.06$& $0.12$& $0.18$& $0.24$ & $0.76$ & $0.82$ & $0.88$ & $0.94$ & $1$ \\
	$[0.008, 0.021]$ & $0$ &$0.06$& $0.12$& $0.18$& $0.70$& $0.76$ & $0.82$ & $0.88$ & $0.94$ & $1$ \\
	$[0.022, 0.045]$ & $0$ &$0.06$& $0.12$& $0.64$& $0.70$& $0.76$ & $0.82$ & $0.88$ & $0.94$ & $1$ \\
	$[0.046, 0.114]$ & $0$ &$0.06$& $0.58$& $0.64$& $0.70$& $0.76$ & $0.82$ & $0.88$ & $0.94$ & $1$ \\
	$[0.115, 1)$ & $0$ & $0.52$& $0.58$& $0.64$& $0.70$& $0.76$ & $0.82$ & $0.88$ & $0.94$ & $1$ \\
\hline
 \end{tabular}
  \caption{IMSE-optimal exact designs in model \eqref{lr} with additional constraints \eqref{const} with respect to the values of the rescaled slope variance $\rho$. The support points of the optimal designs are denoted by $x_1,\ldots,x_{10}$}
  \label{tab2}
\end{table}
		

\bigskip

Due to its continuous nature, the approximate designs in the model with constraints may be easily visualized only for a fixed value of the slope variance. The next picture (Figure~\ref{f3}) presents the optimal approximate design for $\rho=0.1$.
\begin{figure}[ht]
\centering
       \includegraphics[width=100mm]{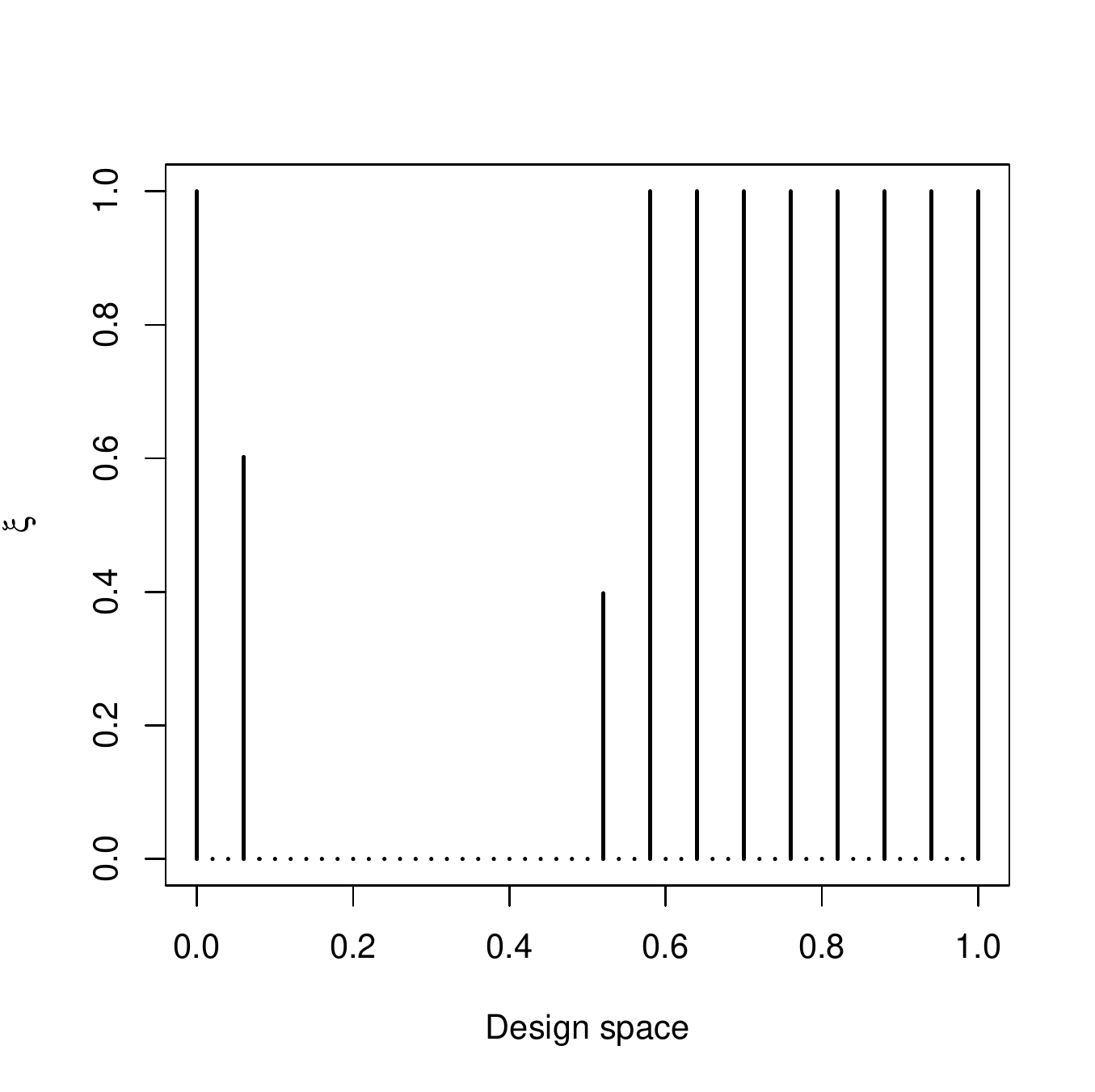}
       \caption{IMSE-optimal approximate designs for $\rho=0.1$ in model \eqref{lr} with additional constraints \eqref{const}}\label{f3} 
    \end{figure} 
	Explicit values for the optimal approximate design are presented in Table~\ref{tab3}.		
\begin{table}
  \centering
\begin{tabular}{cccccccccccc}
\hline
	$x$ & $0$ & $0.06$ & $0.52$& $0.58$& $0.64$& $0.70$& $0.76$ & $0.82$ & $0.88$ & $0.94$ & $1$ \\
 $\xi(x)$  & $1$ & $0.602$& $0.398$& $1$& $1$& $1$& $1$& $1$& $1$ & $1$ & $1$\\
\hline
 \end{tabular}
  \caption{IMSE-optimal approximate designs in model \eqref{lr} with additional constraints \eqref{const} and $\rho=0.1$}
  \label{tab3}
\end{table}		
\bigskip

The computations were performed using the Microsoft Windows 7 Professional operating system with processor Intel(R) Core(TM) i5-6200U CPU at 2.30GHz, 2 cores, RAM 8 GB. The mean computing times for one optimal design (one value of $\rho$) corresponding to the model with standard constraints are $0.2977$ and $4.7918$ seconds for approximate and exact designs, respectively. In the model with additional constraints the corresponding mean computing times are $0.3320$ and $3.4027$ seconds.
\bigskip

Note that the procedure \texttt{od.IQP} (see \cite{HF}), which provides (not necessarily optimal, but usually highly efficient) exact designs, can be used alternatively to \texttt{od.MISOCP}.

\end{document}